\documentclass[letterpaper, 10 pt, conference]{ieeeconf}  

\IEEEoverridecommandlockouts                              

\makeatletter
\def\endthebibliography{%
	\def\@noitemerr{\@latex@warning{Empty `thebibliography' environment}}%
	\endlist
}
\makeatother

\usepackage{cite}
\usepackage{amsmath,amssymb,amsfonts}

\usepackage{graphicx}
\usepackage[export]{adjustbox}
\usepackage{textcomp}

\usepackage{soul}
\usepackage{xcolor}
\definecolor{lightblue}{RGB}{180,220,255}  
\sethlcolor{lightblue}

\usepackage{subcaption}

\usepackage{algpseudocode}
\usepackage{algorithm}

\usepackage{booktabs}

\usepackage{mathtools}

\usepackage{mathrsfs}

\usepackage{todonotes}
\usepackage{xargs}

\newcommandx{\mt}[2][1=]{\todo[linecolor=green,backgroundcolor=green!25,bordercolor=green,author=Mahdi,#1]{#2}}
\newcommand{\inlinecomment}[1]{} 

\newtheorem{remark}{\textbf{Remark}}
\newtheorem{lem}{\textbf{Lemma}}
\newtheorem{theorem}{\textbf{Theorem}}

\newtheorem{assumption}{\textbf{Assumption}}

\renewcommand{\baselinestretch}{0.99}

\title{\LARGE \bf
	Closing the Loop Inside Neural Networks: Causality-Guided Layer Adaptation for Fault Recovery Control}

\author{Mahdi Taheri$^{1}$, Soon-Jo Chung$^{1}$, and Fred Y. Hadaegh$^{1}$
	\thanks{$^{1}$Division of Engineering and Applied Science, California Institute of Technology (Caltech), Pasadena, CA 91125, USA. E-mail:  \{mtaheri, sjchung, hadaegh\}@caltech.edu.}%
	\thanks{
		This research is funded in part by the Technology Innovation Institute and the Defense Advanced Research Projects Agency (Learning Introspective Control).}
}

\begin{document}

	\maketitle
	\thispagestyle{empty}
	\pagestyle{empty}

	\begin{abstract}
		This paper studies the problem of real-time fault recovery control for nonlinear control-affine systems subject to actuator loss of effectiveness faults and external disturbances. We develop a two-stage framework that combines causal inference with selective online adaptation to achieve an effective learning-based recovery control method. In the offline phase, we introduce a causal layer attribution technique based on the average causal effect (ACE) to evaluate the relative importance of each layer in a pretrained deep neural network (DNN) controller compensating for faults. This provides a principled approach to select the most causally influential layer for fault recovery control in the sense of ACE, and goes beyond the widely used last-layer adaptation approach. In the online phase, we deploy a Lyapunov-based gradient update to adapt only the ACE-selected layer to circumvent the need for full-network or last-layer only updates. The proposed adaptive controller guarantees uniform ultimate boundedness (UUB) with exponential convergence of the closed-loop system in the presence of actuator faults and external disturbances. Compared to conventional adaptive DNN controllers with full-network adaptation, our methodology has a reduced computational overhead in the online phase. To demonstrate the effectiveness of our proposed methodology, a case study is provided on a $3$-axis attitude control system of a spacecraft with four reaction wheels. 
				
	\end{abstract}
	
	\section{Introduction}
	Autonomous systems ranging from unmanned aerial vehicles (UAV) to spacecraft are required to have high levels of resilience to operate safely in dynamic and uncertain environments \cite{o2024learning,tsukamoto2024neural}. These systems frequently encounter challenges such as unpredictable external disturbances and various types of faults, which push them into out‑of‑distribution (OOD) conditions beyond their nominal design assumptions \cite{4014460}. Under such circumstances, conventional controllers often fail to recover the system since they lack the level of adaptability required to maintain a reliable performance.

    Adaptive control methods \cite{park2021adaptive,5717148} offer strong stability guarantees when dealing with structured uncertainty or known fault models. However, their effectiveness depends critically on accurate system modeling and prior knowledge of fault characteristics.\inlinecomment{This reliance becomes a significant limitation in real-world deployments, where system degradation and fault signatures may evolve unexpectedly.} Data‑driven controllers that leverage deep neural networks (DNNs) provide a promising alternative by capturing complex mappings directly from data. However, their performance declines when faced with OOD scenarios during runtime~\cite{jiahao2023online}. Hence, online learning approaches where neural network parameters are updated online have been employed as a solution to make a closed-loop autonomous system robust to modeling uncertainties and faults \cite{he2024self,9337905}. However, applying updates across the entire network can result in catastrophic forgetting \cite{rusu2016progressive} and impose extensive computational overhead for real-time implementation.
	
	\subsection{Related Work}
	Spectral normalization, which constrains the Lipschitz constant of each DNN layer by limiting its spectral norm, has been shown to improve generalization under distribution changes by controlling layer-wise sensitivity \cite{miyato2018spectral}. This structural regularization has inspired adaptive strategies that emphasize layer-wise DNN adjustments during online operation. For instance, \cite{o2022neural} utilizes spectral normalization and a meta-learning framework to identify and learn disturbance features offline and update only the final layer of DNN online for rapid residual learning and disturbance rejection. Moreover, in the off-road vehicle domain, \cite{lupu2024magic} proposes an offline meta‑learning method that processes terrain images via a visual foundation model and adapts only the DNN's last layer online to compensate for unseen scenarios during offline training. Progressive neural networks freeze earlier layers and add new ones to continually learn without catastrophic forgetting \cite{rusu2016progressive}.
	
	In the control systems domain, hybrid adaptive methods are utilized to estimate uncertain dynamics via Bayesian learning combined with control‑barrier functions~\cite{dhiman2021control}. In ~\cite{shen2025adaptive}, a long-short-term memory (LSTM)-based identifier whose weights are adapted online under Lyapunov guarantees has been employed. Sparse adaptation techniques utilizing $L_1$ regularization have been proposed to detect actuator faults and reconfigure controllers \cite{o2024learning}. More classical fault-tolerant adaptive control laws update parameters continuously to accommodate various actuator faults \cite{yadegar2021fault}, and reachable set techniques adjust safety envelopes under fault scenarios~\cite{el2023online}. Although these approaches span a wide range of strategies, they treat the DNN-based controller as a single block. In contrast, our framework incorporates causal inference and average causal effect (ACE) to determine which subset of DNN layers is the most critical for fault recovery, and then, considering formal convergence guarantees, it applies adaptive laws specifically to those layers.


    \subsection{Contributions}
    In this paper, first, a nominal controller is designed for the system with no actuator fault and external disturbance. Consequently, we train a DNN controller offline to compensate for a wide range of fault and disturbance scenarios that the nominal controller cannot cope with. However, since our DNN controller may still encounter actuator faults and disturbances for which it was not trained, we introduce a method to identify and adapt only the most causally influential layer of the network. This ensures control recovery effectiveness under OOD conditions. Our methodology operates in two phases. In the offline phase, we simulate various actuator failure and degradation conditions under external disturbances to train the DNN controller in a supervised manner and similar to the imitation learning (IL) approach \cite{tsukamoto2024neural}. Consequently, we compute and use the ACE to measure how changes in each layer's weights affect the closed-loop tracking error. This step allows us to identify which layer has the most significant causal impact on the system's performance under fault conditions. Although the ACE evaluation can be computationally expensive, it is executed offline and before deployment.
    
    In the online phase, only the layer identified by the ACE is updated using an adaptive control law that ensures the tracking error remains uniformly ultimately bounded (UUB) despite unknown faults and disturbances. Subsequently, we investigate and study the case where a fault detection and identification (FDI) module can provide an estimate of the fault.\inlinecomment{ It is shown how utilizing the FDI fault estimates can reduce and improve our error bounds.} This work introduces a novel framework for online learning-based fault recovery control that combines causal inference with exponentially stabilizing online adaptation (see Fig.~\ref{fig:nnUpdate}). The main contributions of the paper are:
	\begin{enumerate}
		\item We introduce a structural causal model (SCM) for the DNN-based controller and develop a new methodology to compute the ACE of each layer on the closed-loop tracking error. This enables one to carry out an offline layer importance evaluation, which will guide future online layer updates.
		\item We develop an adaptation law with exponential stability that updates only the ACE-selected layer, which reduces the computational overhead in the online phase and avoids full-network updates. We prove that this selective adaptation guarantees the UUB of the tracking error in the presence of actuator faults and external disturbances.
		\item We integrate causal inference with adaptive control to identify and adapt an internal DNN layer (beyond the last-layer adaptation in \cite{o2024learning,o2022neural,lupu2024magic}) for real-time fault recovery in nonlinear systems.
	\end{enumerate}

    \begin{figure*}[!t]
		\centering
		\begin{subfigure}[t]{0.48\textwidth}
			\centering
			\includegraphics[width=\linewidth]{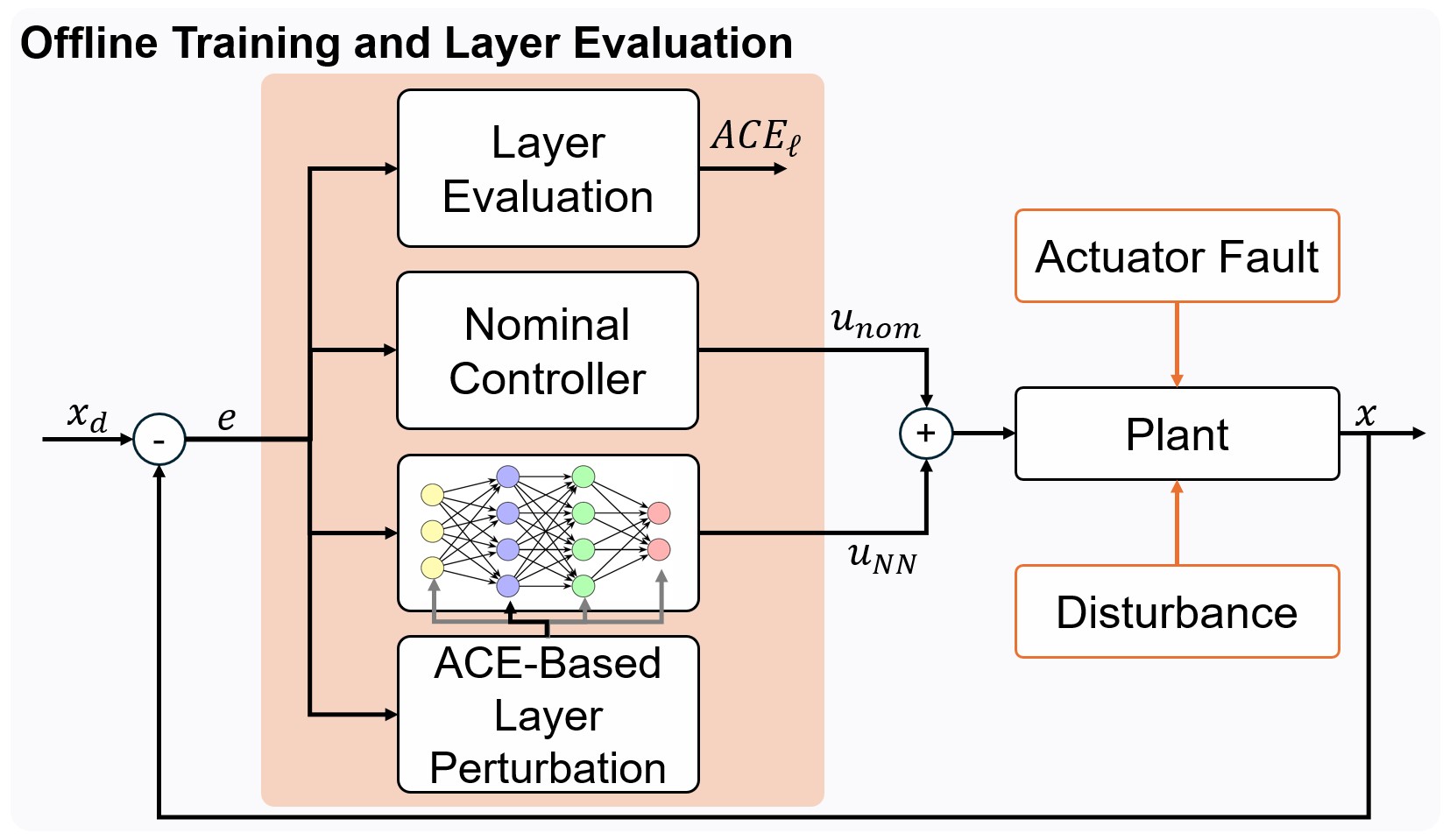}
			\caption{Offline layer‐importance evaluation to select the primary hidden layer for adaptation.}
			\label{fig:nnUpdateOffline}
		\end{subfigure}\hfill
		\begin{subfigure}[t]{0.48\textwidth}
			\centering
			\includegraphics[width=\linewidth]{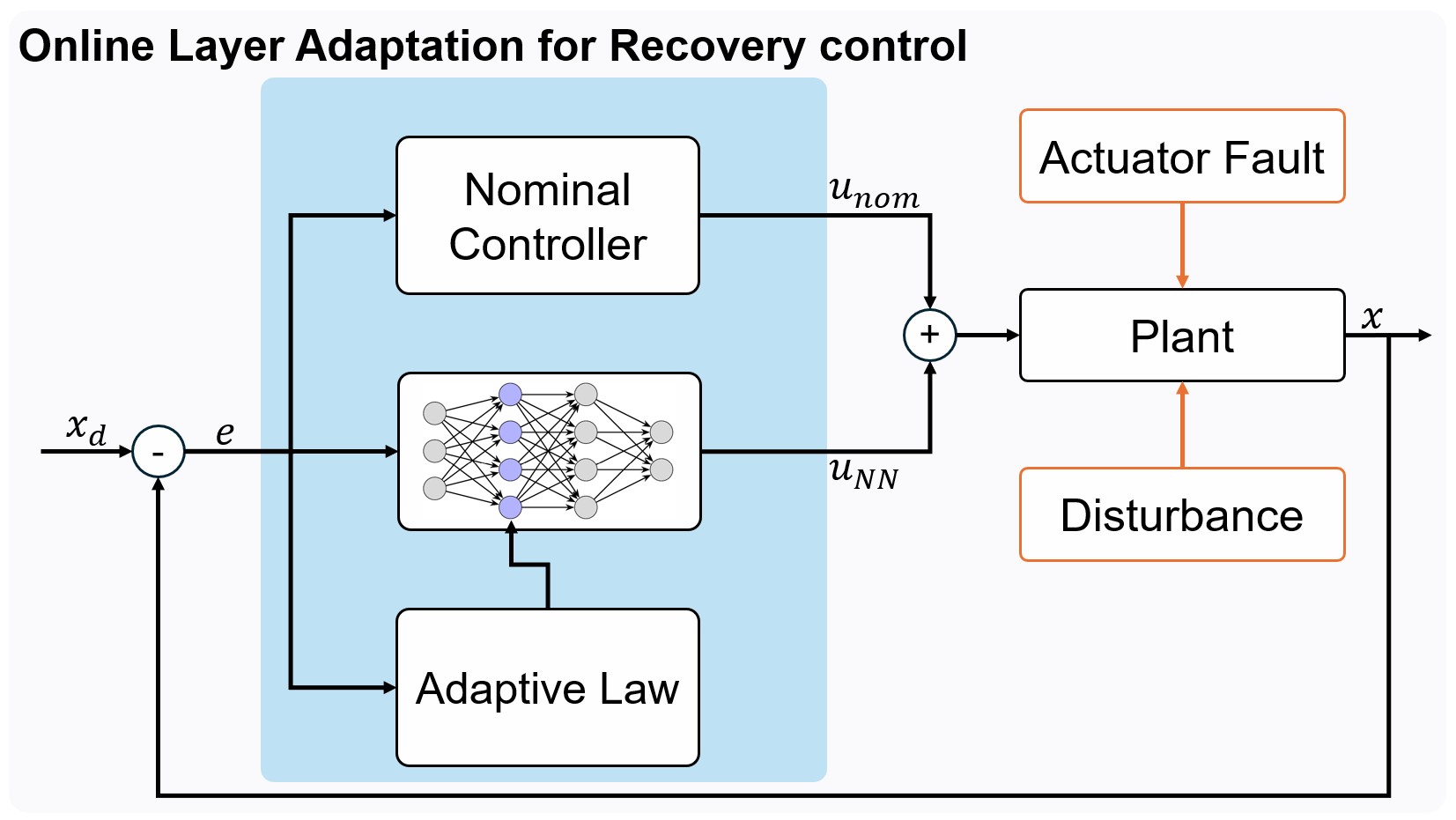}
			\caption{Closed‑loop control architecture showing the adaptive law applied to the selected layer.}
			\label{fig:nnUpdateOnline}
		\end{subfigure}
		\caption{Proposed ACE-guided adaptive fault-recovery framework.}
		\label{fig:nnUpdate}
	\end{figure*}

	\section{Problem Formulation}
	
	\subsection{Notations}
	For a matrix $A$, $\lambda_{\min}(A)$ and $\lambda_{\max}(A)$ denote its minimum and maximum eigenvalues, respectively. The operator $\mathrm{diag}(\cdot)$ forms a diagonal matrix; $\mathrm{tr}(\cdot)$ denotes the trace; $\mathrm{Im}(\cdot)$ is the image (column space); and $A^\dagger$ represents a bounded right-inverse (i.e., the Moore-Penrose pseudo-inverse restricted to $\mathrm{Im}(A)$). The symbol $\|\cdot\|$ denotes the Euclidean $2$-norm; $\|\cdot\|_{\mathrm{F}}$ denotes the Frobenius norm. Finally, $\|\cdot\|_{\mathrm{Lip}}$ represents the global Lipschitz constant of a map.
	
	\subsection{System Model}
	We consider a nonlinear control‐affine system as
	\begin{equation}\label{e:sys}
		\dot{x} = f(x) + g(x)\,\Lambda u + d, 
	\end{equation}
	where $x\in\mathbb{R}^n$ is the state, $u\in\mathbb{R}^m$ is the control input, and $d\in \mathbb{R}^n$ is a bounded disturbance, i.e., $\|d\|\le \bar{d}$, where $\bar{d}>0$ is an unknown upper bound. Also, $f(x)$ and $g(x)$ are known smooth functions. Moreover, $\Lambda=\mathrm{diag}(\eta_1(t),\dots,\eta_m(t))$ is an unknown diagonal matrix, where $0<\eta_k(t)\le1$ denotes a loss of effectiveness fault in the $k$-th actuator, for $k=1,\dots,m$. In this paper, we consider $\eta_\mathrm{min}\le\eta_k\le1$, where $\eta_\mathrm{min}>0$ is an unknown value that indicates the minimum effectiveness of the actuator, and $\eta_k=1$ indicates that the $k$-th actuator is fault-free or healthy. Also, we assume that $d\in\mathrm{Im}(g(x))$, which implies that \eqref{e:sys} has a matched disturbance.
	
	\subsection{Controller Design and Tracking Error}
	Consider \eqref{e:sys} under fault-free and disturbance-free conditions, i.e., $\Lambda=I$ and $d=0$. A nominal control law can be designed in the following form:
	\begin{equation}\label{e:u_nom}
		u_{\mathrm{nom}} = g(x)^\dagger(\dot{x}_d - f(x) - K e),
	\end{equation}
	where $e=x-x_d$ is the tracking error and $x_d\in\mathbb{R}^n$ is the desired trajectory. We choose $(x_d,K)$ so that $\dot x_d - f(x) - K e \in \mathrm{Im}(g(x))$ for all $t$. Moreover, the control gain $K\succ~0$ guarantees $\dot{V_0}(e)\le -\alpha_1(\|e\|)$ under nominal conditions, i.e., $\Lambda=I$ and $d=0$, where $V_0(e)=\frac{1}{2}e^\top Pe$ is a Lyapunov candidate function, $P\succ0$ and symmetric, and $\alpha_1(\cdot)$ is a class-$\mathcal{K}$ function \cite{slotine1991applied}.  
	
	
	Under nominal conditions, having $u=u_{\mathrm{nom}}$ guarantees the exponential stability of the closed-loop system \cite{slotine1991applied}. However, this controller cannot guarantee the boundedness of the tracking error $e$ in the presence of faults and disturbances. Hence, we design a DNN-based controller $u_{\mathrm{NN}}(e,u_{\mathrm{nom}};\bar{W})$ that can compensate for the impact of the actuator fault and the disturbance in the closed-loop system such that the total control input is $u=u_{\mathrm{nom}}+u_{\mathrm{NN}}$, where $\bar{W}$ denotes the design parameters of the DNN.

	\subsection{Problem Statement}
	Given $u_{\mathrm{NN}}$, we identify the update of which DNN layer can result in minimization of tracking error in the presence of unknown faults and disturbances.\inlinecomment{In this step, we compute the ACE of updating each layer of the DNN and identify the layer that should be updated to recover the system under faults and disturbances.} The second problem is to find an adaptive law that can be used to update the weights of the identified layer in the previous step. Moreover, the adaptive law should guarantee the uniform ultimate boundedness of the tracking error.

	\section{Offline Phase: DNN Training and Layer Importance Evaluation}\label{sec:offline}

	\subsection{Data Generation and Supervised DNN Training}\label{ss:train}
	
	A rich dataset of fault and disturbance scenarios is needed to train the $u_{\mathrm{NN}}$ in a supervised manner to develop a baseline fault compensator. Hence, we sample $M$ scenarios $\{\hat{\Lambda}^{(i)},\,d^{(i)}\}_{i=1}^M$, where $\hat{\Lambda}^{(i)}=\mathrm{diag}(\eta_1^{(i)},\dots,\eta_m^{(i)})$ with each $\eta_\mathrm{min}<\eta_k^{(i)}\le 1$ covering light to severe actuator degradation, and $d^{(i)}$ is a bounded disturbance. The $i$-th fault initiation time $t_f^{(i)}$ is selected randomly. Moreover, the desired trajectory $x_d^{(i)}$ can be designed based on our control tracking objectives, and the initial conditions are randomized.
	
	For each scenario $i$, we have
	\begin{equation}\label{e:sys_simu}
		\dot{x}^{(i)} = f(x^{(i)}) + g(x^{(i)})\,\hat{\Lambda}^{(i)}[u_{\mathrm{nom}}^{(i)} + u_{\mathrm{comp}}^{(i)}] + d^{(i)},
	\end{equation}
	where $x^{(i)}\in\mathbb{R}^n$, $u_{\mathrm{nom}}^{(i)}$ is the nominal controller for the $i$-th scenario as per \eqref{e:u_nom}, and $u_{\mathrm{comp}}^{(i)}$ is an ideal controller that compensates for the impact of faults and disturbances. The $u_\mathrm{comp}$ is employed only during the offline phase to generate expert demonstrations that capture actuator fault
	and disturbance effects beyond the nominal controller's capability.
	
	One can rewrite \eqref{e:sys_simu} as $\dot{x}^{(i)} = f(x^{(i)}) + g(x^{(i)})u_{\mathrm{nom}}^{(i)}+g(x^{(i)})u_{\mathrm{comp}}^{(i)} +g(x^{(i)})[\hat{\Lambda}^{(i)}-I]u^{(i)} + d^{(i)}$, where $u^{(i)}=u_{\mathrm{nom}}^{(i)}+u_{\mathrm{comp}}^{(i)}$. Since $d\in \mathrm{Im}(g(x^{(i)}))$, one can always find a certain $\hat{d}^{(i)}$ such that $d^{(i)}=g(x^{(i)})\hat{d}^{(i)}$. Hence, the ideal compensator can be designed as $u_{\mathrm{comp}}^{(i)}=-(\beta^{(i)} u^{(i)} + \hat{d}^{(i)})$, where $\beta^{(i)}=\hat{\Lambda}^{(i)}-I$. Considering that each element of the diagonal matrix $\beta^{(i)}$ is greater than $-1$ and less than or equal to $0$, the matrix $I+\beta^{(i)}$ is invertible. Thus, one obtains $u_{\mathrm{comp}}^{(i)}=-[I+ \beta^{(i)}]^{-1} (u_{\mathrm{nom}}^{(i)} + \hat{d}^{(i)})$. Considering $e^{(i)}=x^{(i)}-x_d^{(i)}$, $u_{\mathrm{nom}}^{(i)}$, and $u_{\mathrm{comp}}^{(i)}$, the training dataset $\mathcal{D}_{\mathrm{train}} = \{(e^{(i)},u_{\mathrm{nom}}^{(i)},u_{\mathrm{comp}}^{(i)})\}_{i=1}^{M}$ can be created.

	Consequently, one can train $u_{\mathrm{NN}}(e,u_{\mathrm{nom}};\bar{W})$ by minimizing the following loss function:
	\begin{align*}
		\mathcal{L}_{\mathrm{train}}(\bar{W})
		= \frac{1}{M} \sum_{i=1}^M \frac{1}{T_i}\int_{0}^{T_i} (\|u_{\mathrm{NN}}-u_{\mathrm{comp}}^{(i)}\|^2 +\lambda_\text{e} \|\hat{e}^{(i)}\|^2),
	\end{align*}
    where $[0,T_i]$ denotes the time window of the $i$-th scenario, and $\dot{\hat{e}}^{(i)}=f(x^{(i)}) + g(x^{(i)})\,\hat{\Lambda}^{(i)}[u_{\mathrm{nom}}^{(i)} + u_{\mathrm{NN}}] + d^{(i)}-\dot{x}_\text{d}^{(i)}$. We also utilize the spectral normalization method \cite{yoshida2017spectral} to enforce a global Lipschitz bound on the $u_\mathrm{NN}$. If the activation function of the DNN has a Lipschitz constant equal to $1$, one obtains $\|u_\mathrm{NN}\|_{\mathrm{Lip}}\le \|W_{1}\cdots W_{L}\bigr\|\le \prod_{\ell=1}^{L}\!\|W_{\ell}\|$, where $W_\ell$ is the $\ell$-th weight matrix of the DNN with appropriate dimensions, for each $\ell=1,\dots,L$. Hence, enforcing $\|W_\ell\|<1$ for every layer results in having $\|u_\mathrm{NN}\|_{\mathrm{Lip}}\le1$.

    The DNN training procedure described above is similar to the behavioral cloning in imitation learning (IL) \cite{argall2009survey,tsukamoto2024neural}. The analytically derived compensator $u_{\mathrm{comp}}$ serves as an expert policy, the collected tuples $(e^{(i)},u_{\mathrm{nom}}^{(i)},u_{\mathrm{comp}}^{(i)})$ constitute the set of demonstrations, and the quadratic loss $\mathcal{L}_{\mathrm{train}}(\bar{W})$ minimizes the error between the learner $u_{\mathrm{NN}}$ and the expert. A distribution shift or mismatch (e.g., unseen fault and disturbance scenarios) can arise when the learner, i.e., $u_{\mathrm{NN}}$, is deployed. Hence, in the next subsection, the ACE-based method is proposed to find the parameters that the learner can update in response to the mentioned distribution shifts.

    Let $\zeta=[u_{\mathrm{nom}}^\top \, e^\top]^\top$. The DNN structure can be shown as $z_{0} =\zeta$, $z_{\ell}=\phi_{\ell}(a_{\ell})$, $a_{\ell}={W_{\ell}} z_{\ell-1}+b_{\ell}$, for all $\ell=1,\dots,L-1$, where $\phi_{\ell}$ is a smooth activation function (e.g., tanh, sigmoid, softplus), $a_{\ell}\in\mathbb{R}^{n_{\ell}}$, $b_{\ell}\in\mathbb{R}^{n_{\ell}}$ is the bias, $z_{\ell-1}\in\mathbb{R}^{n_{\ell-1}}$, and $W_{\ell}\in\mathbb{R}^{n_{\ell}\times n_{\ell-1}}$. The output of the DNN-based controller is	$u_{\mathrm{NN}}=-{W_{L}} z_{L-1}$.
	
	\subsection{Average Causal Effect (ACE) for Layer Importance Evaluation}
	In \cite{chattopadhyay2019neural}, a causal‐attribution framework, which recasts a DNN as a structural causal model (SCM) is proposed to quantify the ACE of any neuron on the output. We adopt this approach and extend it from neuron-level explanations to layer-level importance evaluation for our closed-loop system. In particular, we utilize the $\mathrm{do}$-operator (see \cite{pearl2009causal}) to impose interventions on each weight matrix $W_{\ell}$. These interventions are used to perturb layers of the $u_\mathrm{NN}$ and measure how they alter our closed-loop tracking error $e=x-x_d$. By computing the expected change in $\|e\|$ caused by a small random offset in each layer, one can obtain an ACE score that directly measures how sensitive our error is to a change in layer $\ell$. This causality-based approach offers a principled alternative to gradient-based attribution methods \cite{ancona2018towards} to decide which layer to update and compensate for the impact of faults and disturbances. It should be noted that the this step is carried out offline and before deploying $u_\mathrm{NN}$.
    
	To formalize the above, we represent the DNN-based controller as an SCM. Let the set of nodes $\mathcal{V} =~\{z_0(t),\dots,z_{L}(t),u_{\mathrm{NN}}(t),e(t+\Delta t)\}$ with directed edges in $\mathcal{E}$ that follow the cascade $z_0(t) \to z_1(t) \to \cdots \to z_L(t) \to u_{\mathrm{NN}}(t) \to e(t+\Delta t)$, where $\Delta t>0$ denotes a small time increase in the continuous-time evolution, such that causal ordering is interpreted in a discrete-time unrolled sense and variables at time $t$ influence the state and tracking error at the next instant $t+\Delta t$. In this representation, the weights $W_\ell$ parameterize the edges and interventions are written as $\mathrm{do}(W_\ell \mapsto W_\ell+\Delta_\ell)$. This demonstrates how perturbations propagate from layer $\ell$ to the error $e$.
	
	Let us define the fault, disturbance, and desired trajectory triplet $(\check{\Lambda},\check{d},\check{x}_d)$ that was not used to train the $u_\mathrm{NN}$ in Section~\ref{ss:train}. For each layer $\ell$, we define the ACE as
	\begin{align}\label{e:ace}
		\mathrm{ACE}_{\ell}:= \mathbb{E}[\|\check{e}\|\mid \mathrm{do}(W_{\ell}+\Delta_\ell(t))] -\mathbb{E}[\|\check{e}\|\mid W_{\ell}],
	\end{align}
	where $\mathbb{E}[\|\check{e}\| \mid \mathrm{do}(W_{\ell}+\Delta_\ell(t))]$ is the expected value of the error $\check{e}=\check{x}-\check{x}_d$, and $\check{x}$ is the state generated using \eqref{e:sys} but the actuator fault is set to $\check{\Lambda}$, the disturbance is $\check{d}$, and $\Delta_\ell(t)$ is added to the weight matrix of the $\ell$-th layer of the $u_\mathrm{NN}$. Moreover, $\mathrm{do}(W_{\ell}+\Delta_\ell(t))$ denotes a causal intervention that perturbs layer $\ell$ by a small weight offset $\Delta_\ell(t)$, i.e., the weight of the $\ell$-th layer becomes $W_{\ell}+\Delta_\ell(t)$ (see \cite{pearl2009causal} for more details on the $\mathrm{do}$-operator). We set $\Delta_\ell(t)\sim \mathcal{N} (0,	\rho_\ell^{2}\,I)$, where $\rho_\ell$ is the standard deviation of the perturbation for the $\ell$-th layer. A negative $\mathrm{ACE}_{\ell}$ means that intervening on $\ell$ decreases the tracking error. Thus, the $\ell$-th layer with the smallest $\mathrm{ACE}_{\ell}$ is causally important for compensating for out-of-distribution (OOD) fault and disturbance scenarios. One can also evaluate \eqref{e:ace} for various OOD scenarios of $(\check{\Lambda},\check{d},\check{x}_d)$ and average the resulting $\mathrm{ACE}_{\ell}$ to generalize and improve the ACE-based layer importance evaluation.
	
	Monte Carlo methods are widely used to approximate intractable quantities \cite{gal2016dropout}, such as $\mathrm{ACE}_{\ell}$. Since the exact computation of \eqref{e:ace} is challenging, one can approximate it by using a Monte Carlo estimation as $\widehat{\mathrm{ACE}}_{\ell}	=\frac{1}{N}\sum_{j=1}^N \|e_{j}(W_{\ell}+\Delta_\ell^{(j)})\|-\|e_{j}(W_{\ell})\|$, where we sample $N$ random perturbations $\Delta_\ell^{(j)}$ and simulate the closed‐loop error under the corresponding faults and disturbances. Consequently, we select a layer with the minimum $\widehat{\mathrm{ACE}}_{\ell}$ as the best (in the sense of the ACE) to update in the online step.
	
	Table~\ref{tab:layer_selection_comparison} summarizes various strategies for selecting which DNN layer to adapt during the online step for fault recovery control. These approaches differ in how they quantify layer importance and range from fixed structural rules and local sensitivity metrics to magnitude-based heuristics. In contrast, our proposed ACE-based strategy provides a principled, causal criterion that identifies the layer with the greatest impact (in the sense of \eqref{e:ace}) on the closed-loop tracking error to be updated for fault recovery.

	\begin{table}[t]
		\centering
		\caption{Comparison of layer evaluation strategies.}
		\label{tab:layer_selection_comparison}
		\renewcommand{\arraystretch}{1.2}
		\begin{tabular}{@{}p{2.2cm}p{3.7cm}p{3.cm}@{}}
			\toprule
			\textbf{Method} & \textbf{Description} & \textbf{Main Limitations} \\ 
			\midrule
			
			\textbf{Last‐Layer‐Only} \cite{o2022neural,lupu2024magic} &
			Updates only the final layer of the DNN during online operation.  Assumes early layers retain valid representations across all faults and disturbances. &
			Ignores potential inaccuracy of internal representations; may fail under severe or multiple actuator faults. \\[1mm]
			
			\textbf{Jacobian‐Based Sensitivity} \cite{ancona2018towards}&
			Selects the layer with highest local sensitivity, e.g., $\|\tfrac{\partial u_{\mathrm{NN}}}{\partial W_\ell}\|$ or $\|\tfrac{\partial \mathcal{L}}{\partial W_\ell}\|$.  Measures instantaneous influence of weights on the DNN's output or loss. &
			Captures only short‐term correlations; neglects closed‐loop dynamics and long‐horizon fault effects. \\[1mm]
			
			\textbf{Weight‐Magnitude} \cite{han2015learning} &
			Ranks neurons in each layer based on $\|W_\ell\|$ or $|W_\ell|$ to approximate contribution of each layer. &
			Static and does not reflect how layer perturbations propagate through the system dynamics. \\[1mm]
			
			\textbf{Our ACE‐Based}  &
			Computes $\mathrm{ACE}_{\ell}$ in \eqref{e:ace} to quantify the causal impact of each layer on the closed‐loop error. &
			Requires offline Monte Carlo evaluation. \\		
			
			\bottomrule
		\end{tabular}
	\end{table}

    \begin{remark}
	The definition in \eqref{e:ace} uses $\|\check{e}\|$ as the performance metric. Alternatively, one can adopt a Lyapunov-based measure that ties ACE directly to the closed-loop stability. In particular, we can define $\mathrm{ACE}_{\ell}^V = \mathbb{E} [\dot{V}_0\mid \mathrm{do}(W_{\ell}+\Delta_\ell(t))] - \mathbb{E}[\dot{V}_0 \mid W_\ell]$. This definition of the ACE quantifies how interventions in layer $\ell$ affect the Lyapunov drift and provides a stability-oriented criterion for layer selection.
    \end{remark}

	\section{Online Adaptation of the DNN for Fault Recovery Control}\label{sec:online}
	From Section~\ref{sec:offline}, one can identify which layer in the DNN should be updated online to compensate for actuator faults and disturbances, as shown in Fig.~\ref{fig:nnUpdateOffline}. In this section, we first derive an appropriate adaptive law for $u_{\mathrm{NN}}$, and then show that the aforementioned adaptive law guarantees UUB of the tracking error. In Fig.~\ref{fig:nnUpdateOnline}, the closed‑loop system and adaptive law for the selected layer are depicted. Finally, we show that if there is a fault detection and identification (FDI) module that can provide an estimate of fault values, that is, $\Lambda$, the error bounds would decrease.
	\subsection{Adaptive Law for a Selected Layer}\label{ssec:law}
		Suppose $\ell^*\in\{1,\dots,L\}$ is the selected layer in the previous section to compensate for OOD scenarios. Substituting $u=u_{\mathrm{nom}}+u_{\mathrm{NN}}$ in \eqref{e:sys}, for the error dynamics one obtains
	\begin{equation}\label{e:error}
		\dot{e}=[f(x)+g(x)u_{\mathrm{nom}}-\dot{x}_d]+g(x)(u_{\mathrm{NN}}+r),
	\end{equation}
	where $r=[\Lambda-I]u+\hat{d}$ and $d=g(x)\hat{d}$. 
	\begin{assumption}\label{assum:idealWeight}
		Let $\ell^*\neq L$. There exists a bounded and constant ideal weight $W_{\ell^*}^*=W^*$ such that $r=W_{L} z_{L-1}^*+\epsilon(t)$, where $z_{L-1}^*$ captures the impact of the ideal weight in the $\ell^*$-th layer, i.e., $W_{\ell}^*$, and $\epsilon(t)$ is the bounded approximation error satisfying $\|\epsilon(t)\|\leq \bar{\epsilon}$. If $\ell^*=L$, we assume $r=W^* z_{L-1}+\epsilon(t)$.
	\end{assumption} 
	It should be noted in Assumption~\ref{assum:idealWeight}, $W^*$ is the ideal reference weight toward which the adaptive law drives the estimated weights $W_{\ell^*}$, thereby ensuring minimization of the tracking error within a certain bound (see Theorem~\ref{th:uub}). The existence of such an ideal weight $W^*$ is a standard assumption in adaptive control and neural network–based approximation frameworks, and is commonly employed to establish convergence and boundedness properties \cite{slotine1991applied,o2022neural,lupu2024magic}.
    
	One needs to come up with an adaptive law for the DNN controller such that $u_{\mathrm{NN}} \approx -{W_{L}} z_{L-1}^*$ if $\ell^*\neq L$ and $u_{\mathrm{NN}}~\approx~-{W_{L}^*} z_{L-1}$, otherwise. Therefore, we derive the adaptive law for both cases of $\ell^*\neq L$ and $\ell^*=L$, below.

	Let us consider the virtual loss function $\mathcal{L} = \delta_L^\top u_{\mathrm{NN}}$ for updating the parameters of $u_{\mathrm{NN}}$, where $\delta_L=g(x)^\top P e\in\mathbb{R}^m$. Moreover, for $\ell=L-1,\dots, 1$, we define
	\begin{equation}\label{e:delta}
		\delta_{\ell}=\Psi_{\ell}(a_{\ell})W_{\ell+1}^\top \delta_{\ell+1},
	\end{equation}
	where $\Psi_{\ell}=\mathrm{diag}(\psi_{1,\ell}^\prime(a_{1,\ell}),\dots,\psi_{n_\ell,\ell}^\prime (a_{n_\ell,\ell}))$ is the Jacobian of the activation function. Hence, one obtains
    \begin{equation}\label{e:deltaStar}
        \delta_{\ell^*}=\Psi_{\ell^*} W_{\ell^*+1}^\top \cdots \Psi_{L-1} W_{L}^\top \delta_L,
    \end{equation}
    where $\delta_{\ell^*}$ is the backpropagated error to the $\ell^*$-th layer.
    
    The gradient of the virtual loss function with respect to the adjustable weight matrix is $\nabla_{W_{\ell^{*}}}\mathcal{L}=-\delta_{\ell^*}z_{\ell^*-1}^\top$. Therefore, the adaptive law for updating the layer $\ell^*\in~\{1,\dots,L\}$ can be derived in the following form:
	\begin{equation}\label{e:adaptiveLaw}
		\dot{W}_{\ell^*}=\Gamma \delta_{\ell^*} z_{\ell^*-1}^\top-\gamma W_{\ell^*},
	\end{equation}
     where $\Gamma \succ 0$ is a diagonal gain matrix and $\gamma>0$.

	\subsection{Exponential Stability and Robustness Analysis}
    Let us define $\tilde{W}=W_{\ell^*}-W_{\ell^*}^*$. Consider the case where $\ell^*\neq L$. The first-order Taylor expansion of $z_{L-1}-z_{L-1}^*$ with respect to $W_{\ell^*}$ can be written as
	\begin{align}\label{e:error_z}
		z_{L-1}-z_{L-1}^*=&[\Psi_{L-1} W_{L-1} \Psi_{L-2}\cdots W_{\ell^*+1}\Psi_{\ell^*}]\tilde{W}z_{\ell^*-1} \nonumber\\
		&+\mathcal{O}(\|\tilde{W}\|^2),
	\end{align}
	where $\mathcal{O}(\|\tilde{W}\|^2)$ denotes the higher order terms. Consequently, premultiplying \eqref{e:error_z} by $\delta_L^\top W_L$ yields
	\begin{align}\label{e:error_z2}
		\delta_L^\top W_L S \tilde{W} z_{\ell^*-1}+\mathcal{O}=\mathrm{tr}(\tilde{W}^\top\delta_{\ell^*} z_{\ell^*-1}^\top)+\mathcal{O},
	\end{align}
	where $S=\Psi_{L-1} W_{L-1} \Psi_{L-2}\cdots W_{\ell^*+1}\Psi_{\ell^*} \in \mathbb{R}^{n_{L-1}\times n_{\ell^*}}$.
	
	Consider the following Lyapunov candidate function:
	\begin{equation}\label{e:Lyapunov}
		V(e,\tilde{W})=V_0(e)+\frac{1}{2}\mathrm{tr}(\tilde{W}^\top \Gamma^{-1}\tilde{W}),
	\end{equation} 
	where $V_0(e)=\frac{1}{2}e^\top Pe$. The time derivative of \eqref{e:Lyapunov} can be expressed by $\dot{V}=e^\top P\dot{e}+\mathrm{tr}(\tilde{W}^{\top} \Gamma^{-1}\dot{W}_{\ell^*})$. Considering \eqref{e:error} and $Q=\tfrac{1}{2}(PK + K^\top P^\top)$, one has $\dot{V}=-e^\top Q e+\delta_L^\top(-W_Lz_{L-1}+r)+\mathrm{tr}(\tilde{W}^{\top} \Gamma^{-1}\dot{W}_{\ell^*})$.
	Hence, substituting $r$, \eqref{e:error_z2}, and the adaptive law \eqref{e:adaptiveLaw} yields
	\begin{equation}\label{e:vdot_last}
		\dot{V}= -e^\top Q e -\gamma\mathrm{tr}(\tilde{W}^{\top} \Gamma^{-1} W_{\ell^*})+e^\top P g\epsilon+\mathcal{O}.
	\end{equation}
	
	It should be noted that if $\ell^*=L$, $-W_Lz_{L-1}+r$ above can be recast as $-W_L z_{L-1}+r = -\tilde{W} z_{L-1}+\epsilon$, and the remainder of the derivations for $\dot{V}$ follow along similar lines to the case of $\ell^*\neq L$.

	\begin{assumption}\label{assum:activationFunc}
		For each hidden layer $\ell=1,\dots,L$, the activation function $\phi_{\ell}$ is continuously differentiable (i.e., $\mathcal{C}^1$) and bounded, such that $\|\phi_{\ell}(a_{\ell})\|\leq \bar{\phi}$, where $\bar{\phi}>0$ and the boundedness is guaranteed by spectral normalization. Moreover, the Jacobian of the activation function $\Psi_{\ell}$ satisfies $\|\Psi_{\ell}(a_{\ell})\|\leq \bar{\Psi}$, where $\bar{\Psi}>0$. 
	\end{assumption}
	
	\begin{theorem}\label{th:uub}
		Let Assumptions~\ref{assum:idealWeight} and \ref{assum:activationFunc} hold. Considering the adaptive law \eqref{e:adaptiveLaw}, the tracking error $e$ and weight error $\tilde{W}$ in the closed-loop system \eqref{e:sys} with the total control input $u=u_\mathrm{nom}+u_\mathrm{NN}$ remain UUB with exponential convergence.
	\end{theorem}
	\begin{proof}
		Let $\omega=\|\Gamma^{-\frac{1}{2}}\tilde{W}\|_\mathrm{F}$ and $c=~\|\Gamma^{-\frac{1}{2}}W^*\|_\mathrm{F}$. One has $\gamma\|\Gamma^{-\frac{1}{2}}\tilde{W}\|_\mathrm{F} \, \|\Gamma^{-\frac{1}{2}}W^*\|_\mathrm{F} \leq \frac{\gamma}{2}\omega^2+\frac{\gamma}{2}c^2$. Moreover, Young's inequality results in having
		\begin{equation*}
			\|P\| \, \|g\|  \, \|e\|\bar{\epsilon} \le \frac{\lambda_{\min}(Q)}{2}\|e\|^2+\frac{\|P\|^2 \, \|g\|^2}{2\lambda_{\min}(Q)}\bar{\epsilon}^2.
		\end{equation*}
		Hence, from \eqref{e:vdot_last}, we obtain $\dot{V}\le -\tfrac{\lambda_{\min}(Q)}{2}\|e\|^2-\tfrac{\gamma}{2}\omega^2+\tfrac{\|P\|^2 \, \|g\|^2}{2\lambda_{\min}(Q)}\bar{\epsilon}^2
		+\tfrac{\gamma}{2}c^2+\bar{o}$,	where $\mathcal{O}(\|\tilde{W}\|^2)\le \bar{o}$.
		
		From $V_0\le \tfrac{1}{2}\lambda_{\max}(P)\|e\|^2$, it can be inferred that $-\tfrac{\lambda_{\min}(Q)}{2}\|e\|^2\le -\tfrac{\lambda_{\min}(Q)}{\lambda_{\max}(P)}V_0$. Thus, $\dot{V} \le -\alpha V + \sigma +\bar{o}$, holds for $\alpha=\min\{\tfrac{\lambda_{\min}(Q)}{2\,\lambda_{\max}(P)},\tfrac{\gamma}{2}\}$ and $\sigma = \frac{\|P\|^2\|g\|^2}{2\lambda_{\min}(Q)}\,\bar{\epsilon}^{2} + \frac{\gamma}{2} \|\Gamma^{-\tfrac{1}{2}}W^*\|_\mathrm{F}^2$. Since $\tfrac{\lambda_{\min}(P)}{2}\|e\|^2\le V$, one has
		$\|e(t)\| \le e^{-\tfrac{\alpha t}{2}}\sqrt{\tfrac{2}{\lambda_{\min}(P)} V(0)} +\sqrt{\tfrac{2 \,\sigma+\bar{o}}{\alpha\,\lambda_{\min}(P)}},$ and $\|\Gamma^{-1/2}\tilde{W}(t)\|_{\mathrm{F}} \le e^{-\tfrac{\alpha t}{2}}\sqrt{2 V(0)} +\sqrt{\tfrac{2\,\sigma+\bar{o}}{\alpha}}$.
	\end{proof}
	
	\inlinecomment{It is worth noting that the design of the nominal controller $u_\mathrm{nom}$ does not rely on the DNN controller $u_\mathrm{NN}$. Hence, if there exists a monitoring system that can provide an estimate of the fault severity, i.e., $\hat{\Lambda}$, one can readily reconfigure the nominal controller to reduce the error bounds in Theorem~\ref{th:uub}. The latter has been investigated in the next subsection.}

    \subsection{Improved Error Bounds with a Fault Detection and Identification (FDI) Module}
	Suppose a dedicated FDI module provides an online estimate $\hat{\Lambda}(t)=\mathrm{diag}(\hat{\eta}_1(t),\dots,\hat{\eta}_m(t))$ of the actuator loss of effectiveness matrix $\Lambda(t)=\mathrm{diag}(\eta_1(t),\dots,\eta_m(t))$, where $\hat{\eta}_\mathrm{min}<\hat{\eta}_k\le 1$, and $\hat{\eta}_\mathrm{min}>0$. We reconfigure the nominal controller to account for this estimate as $u_{\mathrm{nom}}^{\mathrm{FDI}} = [g(x)\,\hat{\Lambda}]^\dagger(\dot{x}_d-f(x)-K e)$. Using $u=u_{\mathrm{nom}}^{\mathrm{FDI}}+u_{\mathrm{NN}}$, the error dynamics is
	\begin{align}\label{e:error_fdi}
		\dot{e}= -K e \;+\; g(x)(\Lambda u_{\mathrm{NN}} + [\Lambda-\hat{\Lambda}]\,u_{\mathrm{nom}}^{\mathrm{FDI}} + \hat{d}), 
	\end{align}
	where $d=g(x)\hat{d}$. Considering $g(x)\Lambda u_{\mathrm{NN}} = g(x)\bigl(u_{\mathrm{NN}} + [\Lambda-I]u_{\mathrm{NN}}\bigr)$, one obtains $\dot e = -K e + g(x)(u_{\mathrm{NN}} + r_{\mathrm{FDI}})$,	where $r_{\mathrm{FDI}} = [\Lambda-I]\,u_{\mathrm{NN}} + [\Lambda-\hat{\Lambda}]\,u_{\mathrm{nom}}^{\mathrm{FDI}} + \hat{d}$.

	\begin{assumption}\label{assum:fdi_bound}
		There exist bounded constants $\bar{\Delta}_\Lambda\in(0,1)$, $\kappa_g^\dagger>0$, $\bar{v}>0$, and $\bar{u}_{\mathrm{NN}}>0$ such that $\|\Lambda-\hat{\Lambda}\| \le \bar{\Delta}_\Lambda$, $\|g(x)^\dagger\|\le \kappa_g^\dagger$, $\|v(x)\|\le \bar{v}$, and $\|u_{\mathrm{NN}}\| \le \bar{u}_{\mathrm{NN}}$, $\forall t\ge 0$, where $v(x)=\dot{x}_d-f(x)-K e$.
	\end{assumption}
	
	Let $\|\hat{\Lambda}(t)^{-1}\|\leq \kappa_{\hat{\Lambda}}$, where $\kappa_{\hat{\Lambda}}>0$. Hence, the upper bounds on $u_{\mathrm{nom}}$ and $u_{\mathrm{nom}}^{\mathrm{FDI}}$ can be expressed by $\|u_{\mathrm{nom}}\|\le \kappa_g^\dagger \bar{v}$, and $\|u_{\mathrm{nom}}^{\mathrm{FDI}}\|= \|\,[g\hat{\Lambda}]^\dagger v\|\,\le\,\kappa_{\hat{\Lambda}}\kappa_g^\dagger\,\bar v$.
	
	\begin{lem}\label{lem:residual}
		Under Assumption~\ref{assum:fdi_bound}, one has $\|r\| \le \bar{\beta}(\bar{u}_{\mathrm{NN}}+\kappa_g^\dagger\bar{v})+\|\hat{d}\|=\bar{r}_1$, and $\|r_{\mathrm{FDI}}\| \le \bar{\beta} \,\bar{u}_{\mathrm{NN}}+\bar{\Delta}_\Lambda\,\kappa_{\hat{\Lambda}}\kappa_g^\dagger\bar{v}+\|\hat{d}\|=\bar{r}_2$, for all $t\ge 0$, where $\|\Lambda(t)-I\|\leq \bar{\beta}$ and $\bar{\beta}>0$. Consequently, we obtain
		$\bar{r}_1-\bar{r}_2=\kappa_g^\dagger\bar{v} (\bar{\beta}-\kappa_{\hat{\Lambda}}\bar{\Delta}_\Lambda).$
	\end{lem}
	\begin{proof}
		The proof readily follows from Assumption~\ref{assum:fdi_bound} and definitions of $r$ and $r_{\mathrm{FDI}}$.
	\end{proof}

	From Lemma~\ref{lem:residual}, it can be inferred that the difference between $\bar{r}_1$ and $\bar{r}_2$ depends on the accuracy of the FDI module in the fault estimation, i.e., $\|\Lambda-\hat{\Lambda}\|$, and consequently, its upper bound, i.e., $\bar{\Delta}_\Lambda$. Consider the adaptive law~\eqref{e:adaptiveLaw} and the backpropagated errors~\eqref{e:delta}–\eqref{e:deltaStar}. Similar to Assumption~\ref{assum:idealWeight}, we consider an ideal representation of the residual $r_{\mathrm{FDI}}$ in the following assumption.
	
	\begin{assumption}\label{assum:idealWeightFDI}
		Let $\ell^*\neq L$. There exists a bounded and constant weight $W^*$ such that 
		$r_{\mathrm{FDI}} = W_{L}\,z_{L-1}^* + \epsilon_{\mathrm{FDI}}(t),$
		where $z_{L-1}^*$ is generated by the DNN with $W_{\ell^*}$ replaced by $W^*$, $\|\epsilon_{\mathrm{FDI}}(t)\|\le \bar{\epsilon}_{\mathrm{FDI}}$, and $\bar{\epsilon}_{\mathrm{FDI}}>0$. Moreover, if $\ell^*=L$, we have $r_{\mathrm{FDI}} = W_{L}^*\,z_{L-1} + \epsilon_{\mathrm{FDI}}(t)$.
	\end{assumption} 
	
	\begin{assumption}\label{assum:Lipschitz_r}
		Let $\|r\|\leq \bar{r}_{1}$ and $\|r_\mathrm{FDI}\|\leq\bar{r}_{2}$. Consider $\bar{\epsilon}$ and $\bar{\epsilon}_\mathrm{FDI}$ from Assumptions~\ref{assum:idealWeight} and \ref{assum:idealWeightFDI}, which are the upper bounds on the approximation errors achievable by $u_\mathrm{NN}$ when only the selected layer is adapted. There exists a constant $L_e>0$ such that $\bar{r}_{1}-\bar{r}_{2}\le\,L_e(\bar{\epsilon}-\bar{\epsilon}_\mathrm{FDI})$.
	\end{assumption}
	
	\begin{lem}\label{lem:error_fdi}
		Under Assumptions~\ref{assum:fdi_bound} and \ref{assum:Lipschitz_r}, the approximation error bound of $u_\mathrm{NN}$ satisfies $\bar{\epsilon}_{\mathrm{FDI}}\,\le\,\bar{\epsilon}- \frac{1}{L_e}\,\kappa_g^\dagger\,\bar{v}\,(\bar\beta-\kappa_{\hat{\Lambda}}\bar\Delta_\Lambda)$, and one has a smaller approximation error when the FDI module fault estimation satisfies $\kappa_{\hat{\Lambda}}\bar\Delta_\Lambda<\bar\beta$.
	\end{lem}
	\begin{proof}
		In light of Lemma~\ref{lem:residual}, the residual bound decreases by at least
		$\bar{r}_1-\bar{r}_2= \kappa_g^\dagger\,\bar{v}(\bar{\beta}-\kappa_{\hat{\Lambda}}\bar{\Delta}_\Lambda).$
		Consequently, Assumption~\ref{assum:Lipschitz_r} yields $\bar{\epsilon}_{\mathrm{FDI}} \le\bar{\epsilon}-\frac{1}{L_e}\,(\bar{r}_1-\bar{r}_2)$.
	\end{proof}
	
	Using~\eqref{e:error_fdi}, the same Lyapunov function~\eqref{e:Lyapunov}, and derivations leading to~\eqref{e:vdot_last}, we obtain $\dot{V} = -e^\top Q e - \gamma\,\mathrm{tr}(\tilde{W}^{\top}\Gamma^{-1}W_{\ell^*}) + e^\top P g\,\epsilon_{\mathrm{FDI}} + \mathcal{O}.$
	
	\begin{theorem}\label{th:uub_fdi}
		Let Assumptions~\ref{assum:activationFunc}-\ref{assum:Lipschitz_r} hold. The closed-loop system with $u=u_{\mathrm{nom}}^{\mathrm{FDI}}+u_{\mathrm{NN}}$ remains UUB with exponential convergence. Moreover, if the FDI estimates satisfy $\kappa_{\hat{\Lambda}}\bar{\Delta}_\Lambda<\bar{\beta}$, the upper bound on tracking error $e$ and weight error $\tilde{W}$ will be smaller than the case of $u=u_{\mathrm{nom}}+u_{\mathrm{NN}}$ in Theorem~\ref{th:uub}.
	\end{theorem}
	\begin{proof}
		
		The derivation to show $u=u_{\mathrm{nom}}^{\mathrm{FDI}}+u_{\mathrm{NN}}$ results in having UUB follows along similar lines to that of Theorem~\ref{th:uub}. Moreover, it can be easily seen that the term $\sigma_\mathrm{FDI} = \frac{\|P\|^2\|g\|^2}{2\lambda_{\min}(Q)}\,\bar{\epsilon}_\mathrm{FDI}^{2} + \frac{\gamma}{2} \|\Gamma^{-\tfrac{1}{2}}W^*\|_\mathrm{F}^2$ appears in the upper bounds for $\|e(t)\|$ and $\|\Gamma^{-1/2}\tilde{W}(t)\|_{\mathrm{F}}$. Thus, considering Lemma~\ref{lem:error_fdi}, $\kappa_{\hat{\Lambda}}\bar\Delta_\Lambda<\bar\beta$ results in $\bar{\epsilon}_{\mathrm{FDI}}<\bar{\epsilon}$, which implies smaller upper bounds for errors.
	\end{proof}
	Theorem~\ref{th:uub_fdi} explicitly quantifies how the FDI estimation accuracy, characterized by $\|\Lambda-\hat{\Lambda}\|$ and its upper bound $\bar{\Delta}_\Lambda$, affects the resulting tracking-error bound. Therefore, even in the presence of bounded but imperfect fault estimates, the closed-loop stability is maintained.

	\section{Numerical Case Study}\label{s:simu}
    To validate the proposed fault recovery control methodology, we consider a 3-axis attitude control system with loss of effectiveness faults for a rigid spacecraft that is actuated by four reaction wheels in a tetrahedral configuration. The inertia matrix is given by $I=\mathrm{diag}(1,1,0.8)$ ($\mathrm{kg}\cdot \mathrm{m}^2$), each reaction wheel has an inertia of $J_\omega=0.01$ ($\mathrm{kg}\cdot \mathrm{m}^2$), and the maximum deliverable torque per wheel is limited to $0.14$ ($\mathrm{N}\cdot \mathrm{m}$). We have used \cite{lee2017geometric} to define the spacecraft and reaction wheels parameters. The gains of the nominal controller $u_\mathrm{nom}$ are tuned as $K_p=\mathrm{diag}(22.5,18,15)$ and $K_d=\mathrm{diag}(12,9,7.5)$. Let $\theta=[\theta_1,\, \theta_2, \, \theta_3]^\top \in \mathbb{R}^3$ ($\mathrm{rad}$) where $\theta_1$ denotes the roll, $\theta_2$ is the pitch, and $\theta_3$ is the yaw angle of the spacecraft. The desired roll, pitch, and yaw angles are $\theta_d=[\theta_1^d,\, \theta_2^d,\, \theta_3^d]^\top$, where we have considered $\theta_d(t)=[0.05\sin(0.2\pi t), \, 0.05\cos(0.2\pi t), \frac{\pi}{250} t]^\top$ (i.e., a 3-axis helical trajectory). We define the tracking errors as $\tilde{\theta}_1=\theta_1-\theta^d_1$, $\tilde{\theta}_2=\theta_2-\theta^d_2$, $\tilde{\theta}_3=\theta_3-\theta^d_3$,  and $\tilde{\theta}=\theta-\theta_d$. The trained $u_\mathrm{NN}$ has $6$ layers with the triple $(\tilde{\theta}, \dot{\tilde{\theta}}, u_\mathrm{nom})$ as its input and layers of size $n_1=9$, $n_2=15$, $n_3=15$, $n_4=15$, $n_5=15$, and $n_6=3$ for the body-frame torque. Moreover, by utilizing $\widehat{\mathrm{ACE}}_{\ell}$, we identified the $4$-th layer as the best candidate to update in the recovery control.
	
	\inlinecomment{As can be seen in Fig.~\ref{fig:error_norm_FF}, under fault-free conditions, both controllers $u_\mathrm{nom}$ and $u_\mathrm{nom}+u_\mathrm{NN}$ with the adaptive law \eqref{e:adaptiveLaw} remain below the defined error bound of $0.017$. }We consider $50\%$ and $75\%$ loss of effectiveness faults starting from $t=10$ to $25$ (s) and $t=20$ to $50$ (s) in reaction wheels number $3$ and $2$, respectively.\inlinecomment{ In Fig.~\ref{fig:troquesF}, the effectiveness of each actuator and their applied torques are depicted while the $4$-th layer of $u_\mathrm{NN}$ is being updated as per \eqref{e:adaptiveLaw}.} To demonstrate the impact of updating various layers of the DNN on the tracking error, Fig.~\ref{fig:errorNormF} is provided that illustrates $\|\tilde{\theta}\|$ while the $1$-st, the $4$-th, the $5$-th, the $6$-th (i.e., the last layer), the full-network (i.e., all the layers), and no layer of the network are being updated. As is shown in Fig.~\ref{fig:errorNormF}, updating the $4$-th layer of $u_\mathrm{NN}$ results in having the smallest tracking error. Moreover, results in Table~\ref{tab:rmse_maxerr} show that updating the proposed ACE-selected layer (i.e., the $4$-th layer) achieves the lowest root-mean-square error (RMSE) and maximum error (MaxError) values, with $\mathrm{RMSE} = \sqrt{\frac{1}{T}\!\int_{0}^{T} \|\tilde{\theta}(t)\|^2\,dt}$ and $\mathrm{MaxError} = \max_{t \in [0,T]} \|\tilde{\theta}(t)\|$, for $T>0$. Since the computed ACE values $\widehat{\mathrm{ACE}}_{4}=-0.004387$ and $\widehat{\mathrm{ACE}}_{5}=\widehat{\mathrm{ACE}}_{6}=-0.004256$ are close, it is expected to get similar RMSEs for layers $4$, $5$, and $6$ as shown in Table~\ref{tab:rmse_maxerr}. Finally, we incorporated an FDI module that provides fault estimates subject to estimation errors. As for the FDI module, we employ an Interacting Multiple-Model (IMM) estimator \cite{10104008}, where each filter corresponds to a certain actuator fault severity hypothesis. Any nonlinear state estimator could be used within this IMM structure, but we use Unscented Kalman Filters (UKF). The filters use the control input $u$, $\theta$, and the reaction wheels angular velocities. The IMM computes the posterior probability of each mode, and the most probable mode provides the actuator effectiveness estimate $\hat{\Lambda}$, which satisfies Assumption~\ref{assum:fdi_bound}. Under the same fault scenario a $40\%$ loss of effectiveness was estimated for reaction wheel $3$ (actual $50\%$), and a $45\%$ loss was estimated for reaction wheel $2$ (actual $75\%$). As illustrated in Fig.~\ref{fig:errorNormFwithFDI}, the more accurate estimate for reaction wheel 3 leads to an improvement in tracking performance, whereas the less accurate estimate for reaction wheel~$2$ results in a degradation of performance. As for the computational complexity of the online phase, a forward pass through the DNN costs $O\!\big(\sum_{\ell=1}^{L} n_\ell n_{\ell-1}\big)$. Backpropagating to a single adapted layer $\ell^*$ has cost $O\!\big(\sum_{\ell=\ell^*+1}^{L} n_\ell n_{\ell-1}\big)$, and the update itself adds $O(n_{\ell^*} n_{\ell^*-1})$. In our network $(9,15,15,15,15,3)$, the total number of weights is $1080$. However, in addition to backpropagation through layers $6$ to $4$, the ACE-selected layer~4 only has $225$ parameters to update. 
	
	
	
	\begin{figure}[!t]
    \centering

    \begin{subfigure}{\linewidth}
        \includegraphics[width=\linewidth, height=4.8cm]{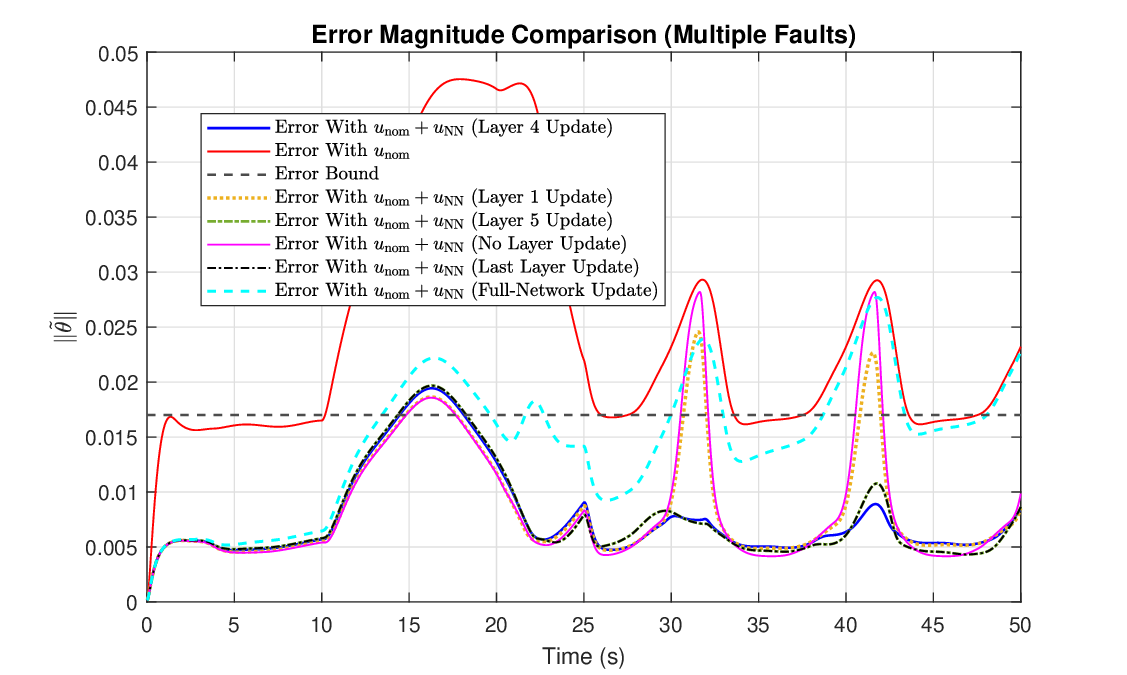}
        \caption{Norm of $\|\tilde{\theta}\|$ (rad) while updating different $u_\mathrm{NN}$ layers.}
        \label{fig:errorNormF}
    \end{subfigure}

    \begin{subfigure}{\linewidth}
        \includegraphics[width=\linewidth, height=4.7cm]{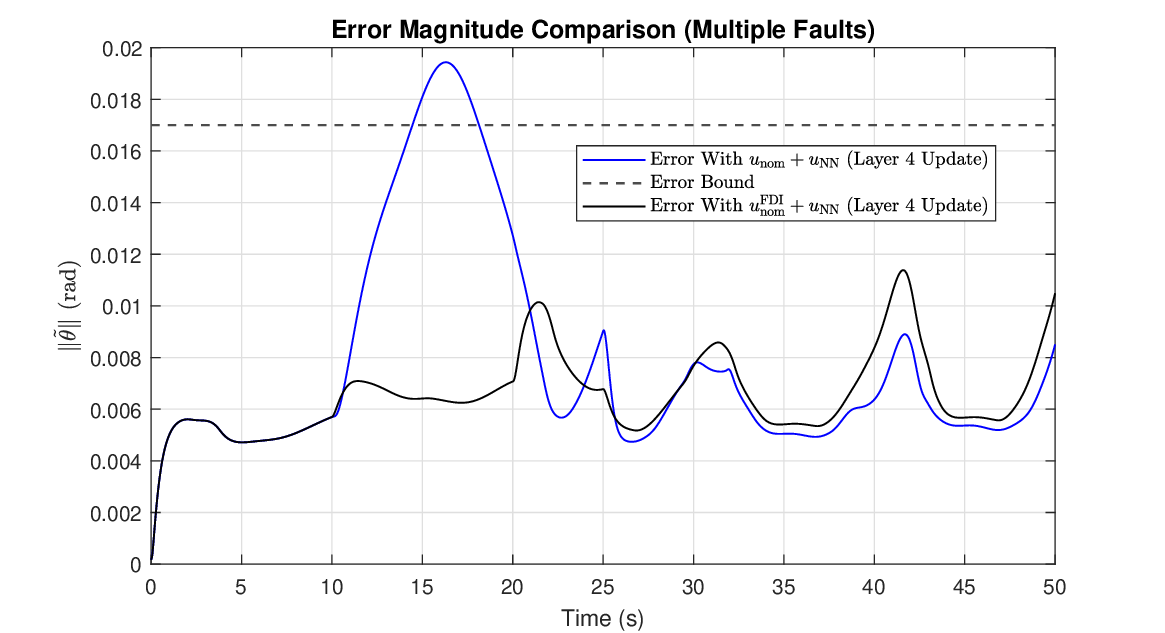}
        \caption{Norm of tracking error $\|\tilde{\theta}\|$ (rad) with FDI.}
        \label{fig:errorNormFwithFDI}
    \end{subfigure}

    \caption{(a) Tracking error with various NN updates; (b) error with FDI reconfiguration.}
    \label{fig:allFaults}
\end{figure}

	\begin{table}[t]
		\centering
		\caption{Comparison of tracking error performance under faults, where lower values indicate better performance.}
		\label{tab:rmse_maxerr}
		\begin{tabular}{lcc}
			\toprule
			\textbf{Adaptation Strategy} & \textbf{RMSE (rad)} & \textbf{Max Error (rad)} \\
			\midrule
			Layer~4 Update (ACE-Selected)                     & 0.008819 & 0.01943 \\
			Layer~6 (i.e., the last layer) Update                 & 0.008899 & 0.01967 \\
			Layer~5 Update                 & 0.008899 & 0.01967 \\
			Layer~1 Update                 & 0.009898 & 0.02456 \\
			No Layer Update                 & 0.010462 & 0.02819 \\
			Full-Network Update                & 0.015695 & 0.02769 \\
			Only Nominal Controller $u_\mathrm{nom}$                   & 0.027062 & 0.04755 \\
			\bottomrule
		\end{tabular}
	\end{table}

	
	\section{Conclusion}\label{s:conclu}
	This paper presented a fault recovery control framework for nonlinear control-affine systems with unknown actuator faults and disturbances. We introduced a two-stage methodology that evaluates the average causal effect (ACE) of each DNN layer to identify the most critical one for fault compensation and adapted only that layer online to ensure uniform ultimate boundedness (UUB) of the tracking error. The proposed causality-guided layer selection goes beyond the widely used last-layer adaptation strategy by revealing which layer is the most impactful for fault recovery control. A spacecraft attitude control case study demonstrated that the proposed ACE-guided adaptation achieved the lowest tracking error under faults. The proposed framework is readily applicable to a broad class of nonlinear control-affine systems. In particular, the offline ACE-based layer selection and the online adaptive law depend only on the closed-loop tracking error under fault and disturbance scenarios. Furthermore, the present work assesses computational feasibility through analytical complexity for hardware implementation.
	
	\bibliographystyle{ieeetr}
	\bibliography{
		Refs/NNUpdateRef}

@book{slotine1991applied,
	title={{Applied Nonlinear Control}},
	author={Slotine, Jean-Jacques E and Li, Weiping and others},
	volume={199},
	number={1},
	year={1991},
	publisher={Prentice hall Englewood Cliffs, NJ}
}

@article{rusu2016progressive,
	title={{Progressive Neural Networks}},
	author={Rusu, Andrei A. and Rabinowitz, Neil C. and Desjardins, Guillaume and Soyer, Hubert and 
	Kirkpatrick, James and Kavukcuoglu, Koray and Pascanu, Razvan and Hadsell, Raia},
	journal={arXiv preprint arXiv:1606.04671},
	year={2016}
}

@article{dhiman2021control,
	title={{Control barriers in Bayesian learning of system dynamics}},
	author={Dhiman, Vikas and Khojasteh, Mohammad Javad and Franceschetti, Massimo and Atanasov, Nikolay},
	journal={IEEE Trans. Autom. Control},
	volume={68},
	number={1},
	pages={214--229},
	year={2021},
	publisher={IEEE}
}

@article{shen2025adaptive,
	title={{Adaptive Control via Lyapunov-Based Deep Long Short-Term Memory Networks}},
	author={Shen, Xuehui and Griffis, Emily J and Wu, Wenyu and Dixon, Warren E},
	journal={IEEE Trans. Autom. Control},
	year={2025},
	publisher={IEEE}
}

@ARTICLE{o2024learning,
  author={O'Connell, Michael and Cho, Joshua and Anderson, Matthew and Chung, Soon-Jo},
  journal={IEEE Robot. Autom. Lett.}, 
  title={Learning-Based Minimally-Sensed Fault-Tolerant Adaptive Flight Control}, 
  year={2024},
  volume={9},
  number={6},
  pages={5198-5205},
  doi={10.1109/LRA.2024.3389414}}

@article{yadegar2021fault,
	title={{Fault-tolerant control of nonlinear heterogeneous multi-agent systems}},
	author={Yadegar, Meysam and Meskin, Nader},
	journal={Automatica},
	volume={127},
	pages={109514},
	year={2021},
	publisher={Elsevier}
}

@article{el2023online,
	title={{Online guaranteed reachable set approximation for systems with changed dynamics and control authority}},
	author={El-Kebir, Hamza and Pirosmanishvili, Ani and Ornik, Melkior},
	journal={IEEE Trans. Autom. Control},
	volume={69},
	number={2},
	pages={726--740},
	year={2023},
	publisher={IEEE}
}

@article{yoshida2017spectral,
	title={{Spectral norm regularization for improving the generalizability of deep learning}},
	author={Yoshida, Yuichi and Miyato, Takeru},
	journal={arXiv preprint arXiv:1705.10941},
	year={2017}
}

@inproceedings{miyato2018spectral,
title={{Spectral Normalization for Generative Adversarial Networks}},
author={Takeru Miyato and Toshiki Kataoka and Masanori Koyama and Yuichi Yoshida},
booktitle={Proc. Int. Conf. Learn. Represent. (ICLR)},
year={2018},
url={https://openreview.net/forum?id=B1QRgziT-},
}

@article{lupu2024magic,
  title={{MAGIC VFM-Meta-Learning Adaptation for Ground Interaction Control With Visual Foundation Models}},
  author={Lupu, Elena Sorina and Xie, Fengze and Preiss, James A and Alindogan, Jedidiah and Anderson, Matthew and Chung, Soon-Jo},
  journal={IEEE Trans. Robot.},
  year={2024},
  publisher={IEEE}
}

@article{o2022neural,
	title={{Neural-fly enables rapid learning for agile flight in strong winds}},
	author={O’Connell, Michael and Shi, Guanya and Shi, Xichen and Azizzadenesheli, Kamyar and Anandkumar, Anima and Yue, Yisong and Chung, Soon-Jo},
	journal={Science Robotics},
	volume={7},
	number={66},
	pages={65-97},
	year={2022},
	publisher={American Association for the Advancement of Science}
}

@article{argall2009survey,
  title={{A survey of robot learning from demonstration}},
  author={Argall, Brenna D and Chernova, Sonia and Veloso, Manuela and Browning, Brett},
  journal={Robotics and Autonomous Systems},
  volume={57},
  number={5},
  pages={469--483},
  year={2009},
  publisher={Elsevier}
}

@article{tsukamoto2024neural,
  title={{Neural-Rendezvous: Provably Robust Guidance and Control to Encounter Interstellar Objects}},
  author={Tsukamoto, Hiroyasu and Chung, Soon-Jo and Nakka, Yashwanth Kumar and Donitz, Benjamin and Mages, Declan and Ingham, Michel},
  journal={J. Guid. Control Dyn.},
  volume={47},
  number={12},
  pages={2525--2543},
  year={2024},
  publisher={American Institute of Aeronautics and Astronautics}
}

@inproceedings{chattopadhyay2019neural,
  title={{Neural network attributions: a causal perspective}},
  author={Chattopadhyay, Aditya and Manupriya, Piyushi and Sarkar, Anirban and Balasubramanian, Vineeth N},
  booktitle={Int. Conf. Mach. Learn.},
  pages={981--990},
  year={2019},
  organization={PMLR}
}

@inproceedings{ancona2018towards,
title={{Towards better understanding of gradient-based attribution methods for Deep Neural Networks}},
author={Marco Ancona and Enea Ceolini and Cengiz Öztireli and Markus Gross},
booktitle={Int. Conf. Learn. Rep.},
year={2018},
url={https://openreview.net/forum?id=Sy21R9JAW},
}

@article{pearl2009causal,
  author = {Pearl, Judea},
  title = {{Causal inference in statistics: An overview}},
  journal = {Statistics Surveys},
  volume = {3},
  pages = {96--146},
  year = {2009},
  issn = {1935-7516},
  doi = {10.1214/09-SS057},
}

@inproceedings{gal2016dropout,
  title={{Dropout as a Bayesian approximation: Representing model uncertainty in deep learning}},
  author={Gal, Yarin and Ghahramani, Zoubin},
  booktitle={Int. Conf. Mach. Learn.},
  pages={1050--1059},
  year={2016},
  organization={PMLR}
}

@article{han2015learning,
	title={Learning both weights and connections for efficient neural network},
	author={Han, Song and Pool, Jeff and Tran, John and Dally, William},
	journal={Adv. Neural Inf. Process. Syst.},
	volume={28},
	year={2015}
}

@ARTICLE{4014460,
  author={Tafazoli, S. and Khorasani, K.},
  journal={IEEE Trans. Aerosp. Electron. Syst.}, 
  title={{Nonlinear control and stability analysis of spacecraft attitude recovery}}, 
  year={2006},
  volume={42},
  number={3},
  pages={825-845},
  doi={10.1109/TAES.2006.248187}}

@article{park2021adaptive,
	title={{Adaptive Fault-Tolerant Flight Control for Input-Redundant Systems Using a Nonlinear Reference Model}},
	author={Park, Hyunsang and Kim, Youdan},
	journal={IEEE Trans. Aerosp. Electron. Syst.},
	volume={57},
	number={5},
	pages={3337--3356},
	year={2021}
}

@INPROCEEDINGS{5717148,
  author={Chowdhary, Girish and Johnson, Eric},
  booktitle={Proc. IEEE Conf. Decis. Control}, 
  title={{Concurrent learning for convergence in adaptive control without persistency of excitation}}, 
  year={2010},
  volume={},
  number={},
  pages={3674-3679},
  doi={10.1109/CDC.2010.5717148}}

@inproceedings{jiahao2023online,
  title={{Online dynamics learning for predictive control with an application to aerial robots}},
  author={Jiahao, Tom Z and Chee, Kong Yao and Hsieh, M Ani},
  booktitle={Proc. Conf. Robot Learn.},
  pages={2251--2261},
  year={2023},
  organization={PMLR}
}

@ARTICLE{9337905,
  author={Sun, Runhan and Greene, Max L. and Le, Duc M. and Bell, Zachary I. and Chowdhary, Girish and Dixon, Warren E.},
  journal={IEEE Control Syst. Lett.}, 
  title={{Lyapunov-Based Real-Time and Iterative Adjustment of Deep Neural Networks}}, 
  year={2022},
  volume={6},
  number={},
  pages={193-198},
  doi={10.1109/LCSYS.2021.3055454}}

@INPROCEEDINGS{he2024self,
	author={He, Guanqi and Choudhary, Yogita and Shi, Guanya},
	booktitle={Proc. IEEE Int. Conf. Robot. Autom.}, 
	title={{Self-Supervised Meta-Learning for All-Layer DNN-Based Adaptive Control with Stability Guarantees}}, 
	year={2025},
	pages={6012--6018},
}

@article{lee2017geometric,
  title={{Geometric mechanics based nonlinear model predictive spacecraft attitude control with reaction wheels}},
  author={Lee, Dae Young and Gupta, Rohit and Kalabi{\'c}, Uro{\v{s}} V and Di Cairano, Stefano and Bloch, Anthony M and Cutler, James W and Kolmanovsky, Ilya V},
  journal={J. Guid. Control Dyn.},
  volume={40},
  number={2},
  pages={309--319},
  year={2017},
}

@INPROCEEDINGS{10104008,
	author={Taheri, Mahdi and Nematollahi, Mohammadreza and Khorasani, Khashayar},
	booktitle={IEEE Int. Symp. Inertial Sensors Syst.}, 
	title={Detection and Identification of GNSS Spoofing Cyber-Attacks for Naval Marine Vessels}, 
	year={2023},
	volume={},
	number={},
	pages={1-4},
	doi={10.1109/INERTIAL56358.2023.10104008}}
\end{document}